\documentclass{fundam}

\publyear{25}
\papernumber{10}
\volume{195}
\issue{1-4}
\theDOI{10.46298/fi.12365}

\versionForARXIV

\usepackage{url} %
\usepackage{graphicx}%

\usepackage{listings}
\makeatletter
\def\lst@lettertrue{\let\lst@ifletter\iffalse}
\makeatother

\usepackage{amsmath,amsfonts}
\usepackage{mathtools}

\usepackage{graphviz}
\usepackage{svg}

\usepackage{xparse}
\usepackage{ifthen}
\usepackage{hyperref}

\usepackage{cleveref}

\usepackage{algorithm, algpseudocode}

\usepackage{xcolor}
\definecolor[named]{ACMBlue}{cmyk}{1,0.1,0,0.1}
\definecolor[named]{ACMYellow}{cmyk}{0,0.16,1,0}
\definecolor[named]{ACMOrange}{cmyk}{0,0.42,1,0.01}
\definecolor[named]{ACMRed}{cmyk}{0,0.90,0.86,0}
\definecolor[named]{ACMLightBlue}{cmyk}{0.49,0.01,0,0}
\definecolor[named]{ACMGreen}{cmyk}{0.20,0,1,0.19}
\definecolor[named]{ACMPurple}{cmyk}{0.55,1,0,0.15}
\definecolor[named]{ACMDarkBlue}{cmyk}{1,0.58,0,0.21}
\hypersetup{colorlinks,
    linkcolor=ACMPurple,
    citecolor=ACMPurple,
    urlcolor=ACMDarkBlue,
    filecolor=ACMDarkBlue}
\usepackage{hyperref}

\usepackage{rotating}

\usepackage{diagbox}

\MakeRobust{\ref}%

\makeatletter
\newcommand{\labeltext}[2]{%
    \@bsphack
    \csname phantomsection\endcsname %
    \def\@currentlabel{#1}{\label{#2}}%
    \@esphack
}
\makeatother

\usepackage[xcolor, notion]{knowledge}
\knowledgeconfigureenvironment{theorem,lemma,proof}{}

\usepackage{tikz}
\usetikzlibrary{calc, trees, bending, backgrounds, shapes, shapes.geometric, fit, arrows, arrows.meta, positioning, graphs, quotes}

\usepackage{ebproof}

\usepackage{stackengine}

\makeatletter
\let\orgdescriptionlabel\descriptionlabel
\renewcommand*{\descriptionlabel}[1]{%
    \let\orglabel\label
    \let\label\@gobble
    \phantomsection
    \edef\@currentlabel{#1}%
    \let\label\orglabel
    \orgdescriptionlabel{#1}%
}
\makeatother

\usepackage{centernot}

\usepackage[notion]{knowledge}
\newcommand{\set}[1]{\{ #1 \}}
\newcommand{\tuple}[1]{\langle #1 \rangle}
\newcommand{\const}[1]{\mathsf{#1}}
\newcommand{\bl}{\_}
\newcommand{\defeq}{\mathrel{\ensurestackMath{\stackon[1pt]{=}{\scriptscriptstyle\Delta}}}}
\newcommand{\defiff}{\mathrel{\ensurestackMath{\stackon[1pt]{\iff}{\scriptscriptstyle\Delta}}}}

\newcommand{\nat}{\mathbb{N}}

\newcommand{\domain}[1]{|#1|}

\usepackage{amssymb}

\usepackage{stmaryrd}
\newcommand{\jump}[1]{\llbracket#1\rrbracket}

\NewDocumentCommand\model{O{1}}{%
    \ifcase#1
        undefined
    \or \mathfrak{A}
    \or \mathfrak{B}
    \else undefined

    \fi
}

\NewDocumentCommand\fml{O{1}}{%
    \ifcase#1
        undefined
    \or \varphi
    \or \psi
    \or \rho
    \else undefined

    \fi
}

\NewDocumentCommand\term{O{1}}{%
    \ifcase#1
        undefined
    \or t
    \or s
    \or u
    \else undefined

    \fi
}

\tikzstyle{mynode} = [inner sep = 1.5pt, fill= gray!20]
\tikzstyle{mysmallnode} = [inner sep = 1.pt, fill= gray!20]
\tikzset{earrow/.style={>={{[flex] Latex[length=.1cm, width=2.5pt]}}}}

\knowledge{ignore}
| validity@pl

\knowledge{notion}
| semantic equivalence w.r.t.\  binary relations
| semantically equivalent w.r.t.\  binary relations

\knowledge{notion}
| semantic equivalence
| semantic equivalences
| semantically equivalent

\knowledge{notion}
| valid
| validity

\knowledge{notion}
| finitely valid
| finite validity

\knowledge{notion}
| conservative reduction

\knowledge{notion}
| structure
| structures

\knowledge{notion}
| finite@structure
| finiteness@structure

\knowledge{notion}
| isomorphic@structure
| isomorphism@structure

\knowledge{notion}
| isomorphism closure@structure

\knowledge{notion}
| first-order predicate logic with equality

\knowledge{notion}
| variables
| variable
| variables@fo
| variable@fo

\knowledge{notion}
| formulas
| formula
| formulas@fo
| formula@fo

\knowledge{notion}
| sentences
| sentence
| sentences@fo
| sentence@fo

\knowledge{notion}
| size@fo

\knowledge{notion}
| semantics@fo

\knowledge{notion}
| true@fo

\knowledge{notion}
| satisfiable@fo
| satisfiability@fo

\knowledge{notion}
| finitely satisfiable@fo
| finite satisfiability@fo

\knowledge{notion}
| The calculus of relations
| the calculus of relations

\knowledge{notion}
| terms
| term
| terms@cor
| term@cor

\knowledge{notion}
| size@cor

\knowledge{notion}
| semantics@cor

\knowledge{notion}
| semantic relation@cor formula

\knowledge{notion}
| true@cor

\knowledge{notion}
| quantifier-free formulas
| quantifier-free formula
| quantifier-free formulas@cor
| quantifier-free formula@cor
| formula@cor
| formulas@cor

\knowledge{notion}
| semantics@cor formula

\knowledge{notion}
| size@cor formula

\knowledge{notion}
| equation
| equations
| equation@cor
| equations@cor

\knowledge{notion}
| inequation
| inequations
| inequation@cor
| inequations@cor

\knowledge{notion}
| (term) variable@cor
| (term) variables@cor
| variable@cor
| variables@cor

\knowledge{notion}
| $k$-tuple structure
| $k$-tuple structures

\knowledge{notion}
| dot-dagger alternation hierarchy

\knowledge{notion}
| Schr{\"o}der-Tarski translation

\knowledge{notion}
| standard translation
 
\begin{document}

\title{Note on a Translation from First-Order Logic into the Calculus of Relations Preserving Validity and Finite Validity}

\address{nakamura.yoshiki.ny{@}gmail.com}

\author{Yoshiki Nakamura\thanks{We would like to thank the anonymous reviewers for their useful comments.
This work was supported by JSPS KAKENHI Grant Number JP21K13828.}\\%
  Department of Computer Science\\
  Institute of Science Tokyo\\
  nakamura.yoshiki.ny{@}gmail.com
}

\maketitle

\runninghead{Y. Nakamura}{Note on a translation from FO into CoR preserving vailidity and finite vailidity}

\begin{abstract}
In this note, we give a linear-size translation from formulas of first-order logic into equations of the calculus of relations preserving validity and finite validity.
Our translation also gives a linear-size conservative reduction from formulas of first-order logic into formulas of the three-variable fragment of first-order logic.
 \end{abstract}

\begin{keywords}
  first-order logic, relation algebra
\end{keywords}

\section{Introduction}
\emph{\kl{The calculus of relations}} (CoR, for short) \cite{tarskiCalculusRelations1941} is an algebraic system with operations on binary relations.
As binary relations appear everywhere in computer science, CoR and relation algebras can be applied to various areas, such as databases and program development and verification \cite{givantCalculusRelations2017}.
W.r.t.\ binary relations, CoR has the same expressive power as the three-variable fragment of \kl{first-order predicate logic with equality} ($\mathrm{FO3}_{=}$) (where all predicate symbols are binary) \cite{tarskiFormalizationSetTheory1987}, so CoR has strictly less expressive power than \kl{first-order predicate logic with equality} ($\mathrm{FO}_{=}$).
For example, CoR \kl{equations} cannot characterize the class of structures s.t.\ ``its cardinality is greater than or equal to $4$'', whereas it can be characterized by the $\mathrm{FO}_{=}$ \kl{formula} $\forall x_1, \forall x_2, \forall x_3, \exists y, (\lnot {y = x_1}) \land (\lnot {y = x_2}) \land (\lnot {y = x_3})$ where $x_1, x_2, x_3, y$ are pairwise distinct variables.

Nevertheless, there is a recursive translation (total recursive function) from $\mathrm{FO}_{=}$ \kl{formulas} into CoR \kl{equations} (resp.\ $\mathrm{FO3}_{=}$ \kl{formulas}) preserving \kl{validity} \cite{tarskiFormalizationSetTheory1987} (see also \cite{madduxFinitaryAlgebraicLogic1989,andrekaReducingFirstorderLogic2013}).

In this paper, we give another recursive translation from $\mathrm{FO}_{=}$ \kl{formulas} into CoR \kl{equations} preserving \kl{validity}, slightly refined in that it satisfies both of the following:
\begin{enumerate}
  \item Our translation preserves both \kl{validity} and \kl{finite validity} (so, it also gives a \intro*\kl{conservative reduction} \cite[Def.\ 2.1.35]{borgerClassicalDecisionProblem1997}).
  \item Our translation is linear-size (i.e., the output size is bounded by a linear function in the input size).
\end{enumerate}
The first refinement is useful, e.g., in finding counter-models (because if there exists a finite counter-model in the pre-translated \kl(fo){formula}, then there also exists a finite counter-model in the post-translated \kl(cor){formula}).
Such a translation is already known (e.g., \cite[Cor.\ 3.1.8 and Thm.\  3.1.9]{borgerClassicalDecisionProblem1997}), but via encodings of Turing-machines and domino problems.
Our translation presents a \kl{conservative reduction} from $\mathrm{FO}_{=}$ \kl{formulas} to $\mathrm{FO3}_{=}$ \kl{formulas}, directly.
Thanks to this, we also have the second refinement, which shows that the \kl{validity} (resp.\ \kl{finite validity}) problem of $\mathrm{FO}_{=}$ \kl{formulas} and that of CoR \kl{equations} are equivalent under linear-size translations, as the converse direction immediately follows from the \kl{standard translation} from CoR \kl{equations} into $\mathrm{FO3}_{=}$ \cite{tarskiCalculusRelations1941} (Prop.\ \ref{proposition: standard translation}).

Our translation is not so far from known encodings (e.g., \cite{madduxFinitaryAlgebraicLogic1989, andrekaReducingFirstorderLogic2013}) in that they and our translation use pairing ($2$-tupling) functions, but in our translation, we use \emph{non-nested} $k$-tupling functions where $k$ is an arbitrary natural number, instead of arbitrarily nested pairing functions.
For constructions using arbitrarily nested pairing functions, we need infinitely many vertices even if the base universe is finite (as there is no surjective function from $X$ to $X^2$ when $\# X$ is finite and $\# X \ge 2$).
Thanks to the modification above, our construction preserves both \kl{validity} and \kl{finite validity}.
Additionally, to preserve the output size linear in the input size, we apply a cumulative sum technique.

This paper is structured as follows.
In \Cref{section: preliminaries}, we give basic definitions of $\mathrm{FO}_{=}$ and CoR.
In \Cref{section: reduction}, we give a translation from $\mathrm{FO}_{=}$ \kl{formulas} into CoR \kl{equations} preserving \kl{validity} and \kl{finite validity}.
In \Cref{section: Tseitin}, we additionally give a Tseitin translation for CoR, which is useful for reducing the number of alternations of operations.
Additionally, in \Cref{section: direct proof}, we give a direct translation from $\mathrm{FO}_{=}$ \kl{formulas} into $\mathrm{FO3}_{=}$ \kl{formulas}, not via CoR, for explicitly writing a transformed $\mathrm{FO3}_{=}$ \kl{formulas}; the translation is the same as that given in \Cref{section: reduction}.

\newcommand{\sig}{\Sigma}
\newcommand{\V}{\mathbf{V}}
\section{Preliminaries}\label{section: preliminaries}
We write $\nat$ for the set of all non-negative integers.
For a set $A$, we write $\# A$ for the cardinality of $A$.

A \intro*\kl{structure} $\model$ over a set $A$ is a tuple $\tuple{\domain{\model}, \set{a^{\model}}_{a \in A}}$, where
\begin{itemize}
  \item the universe $\domain{\model}$ is a non-empty set of vertices,
  \item each $a^{\model} \subseteq \domain{\model}^{2}$ is a binary relation on $\domain{\model}$.
\end{itemize}
We say that a \kl{structure} $\model$ is \intro*\kl(structure){finite} if $\domain{\model}$ is finite.
For \kl{structures} $\model[1], \model[2]$ over a set $A$,
we say that $\model[1]$ and $\model[2]$ are \intro*\kl(structure){isomorphic} if there is a bijective map $f \colon \domain{\model[1]} \to \domain{\model[2]}$ such that for all $x, y \in \domain{\model}$ and $a \in A$, we have $\tuple{x, y} \in a^{\model[1]} \Longleftrightarrow \tuple{f(x), f(y)} \in a^{\model[2]}$.

\subsection{First-order logic}\label{section: first-order logic}
Let $\sig$ be a countably infinite set of binary predicate symbols and $\V$ be a countably infinite set of \intro*\kl(fo){variables}.
The set of \intro*\kl(fo){formulas} in \intro*\kl{first-order predicate logic with equality} ($\mathrm{FO}_{=}$) is defined by:
\begin{align*}
  \fml[1], \fml[2], \fml[3] & \quad::=\quad a(x, y) \mid {x = y} \mid \lnot \fml[2] \mid \fml[2] \land \fml[3] \mid \exists x, \fml[2] \tag{$a \in \sig$ and $x, y \in \V$}
\end{align*}
We write $\mathrm{V}(\fml)$ for the set of free and bound \kl{variables} occurring in a \kl{formula} $\fml$.
For $k \ge 0$, we write $\mathrm{FO}k_{=}$ for the set of all \kl{formulas} $\fml$ s.t.\ $\# \mathrm{V}(\fml) \le k$.
(In this paper, $\mathrm{FO}3_{=}$ mostly occurs.)
A \intro*\kl(fo){sentence} is a \kl(fo){formula} not having any free \kl{variable}.
We use parentheses in ambiguous situations and use the following notations:
\begin{align*}
  \fml[1] \lor \fml[2] & \;\defeq\; \lnot ((\lnot \fml[1]) \land (\lnot \fml[2])) & 
  \fml[1] \to \fml[2] & \;\defeq\; (\lnot \fml[1]) \lor \fml[2] &
  \const{t}            & \;\defeq\; \exists x, {x = x}                            \\
   \forall x, \fml[2] & \;\defeq\; \lnot \exists x, \lnot \fml[2] &
  \fml[1] \leftrightarrow \fml[2] & \;\defeq\; (\fml[1] \to \fml[2]) \land (\fml[2] \to \fml[1]) &
   \const{f}          & \;\defeq\; \lnot \const{t}
\end{align*}
We write $\bigwedge \Gamma$ for the \kl(fo){formula} $\fml_1 \land \dots \land \fml_n$
where $\Gamma = \set{\fml_1, \dots, \fml_n}$ is a finite set (and $\fml_1, \dots, \fml_n$ are ordered by a total order).
The \intro*\kl(fo){size} $\|\fml\| \in \nat$ of a \kl(fo){formula} $\fml$ is defined by:
\begin{align*}
  \|a(x, y)\|               & \defeq 1 + 2                         &
  \|x = y\|                 & \defeq 1 + 2                         &
  \|\lnot \fml[2]\|         & \defeq 1 + \|\fml[2]\|                 \\
  \|\fml[2] \land \fml[3]\| & \defeq 1 + \|\fml[2]\| + \|\fml[3]\| &
  \|\exists x, \fml[2]\|    & \defeq 1 + 1 + \|\fml[2]\|
\end{align*}

For a \kl{structure} $\model$ over $\sig$,
the \intro*\kl(fo){semantics} $\jump{\fml}^{\model} \subseteq \domain{\model}^{\V}$ of a \kl(fo){formula} $\fml$ over $\model$ is defined as follows (where $\domain{\model}^{\V}$ denotes the set of functions from $\V$ to $\domain{\model}$):
\begin{align*}
  \jump{a(x, y)}^{\model}               & \defeq \set{f \colon \V \to \domain{\model} \mid \tuple{f(x), f(y)} \in a^{\model}}                                              \\
  \jump{x = y}^{\model}                 & \defeq \set{f \colon \V \to \domain{\model} \mid f(x) = f(y)}                                                                    \\
  \jump{\lnot \fml[2]}^{\model}         & \defeq \domain{\model}^{\V} \setminus \jump{\fml[2]}^{\model}                                                                    \\
  \jump{\fml[2] \land \fml[3]}^{\model} & \defeq \jump{\fml[2]}^{\model} \cap \jump{\fml[3]}^{\model}                                                                      \\
  \jump{\exists x, \fml[2]}^{\model}    & \defeq \set{f \colon \V \to \domain{\model} \mid \mbox{ for some $v \in \domain{\model}$, $f[v/x] \in \jump{\fml[2]}^{\model}$}}
\end{align*}
Here, $f[v/x]$ denotes the function $f$ in which the value of $x$ has been replaced with $v$.

For a \kl(fo){formula} $\fml$ and a \kl{structure} $\model$,
we say that $\fml$ is \intro*\kl(fo){true} on $\model$, written $\model \models \fml$, if $\jump{\fml}^{\model} = \domain{\model}^{\V}$.
We say that a \kl(fo){formula} $\fml$ is \intro*\kl{valid} (resp.\ \intro*\kl{finitely valid}) if
$\jump{\fml}^{\model} = \domain{\model}^{\V}$ holds for all \kl{structures} (resp.\ all \kl(structure){finite} \kl{structures}) $\model$.
We say that two \kl(fo){formulas} $\fml[1], \fml[2]$ are \intro*\kl{semantically equivalent} if the \kl(fo){formula} $\fml[1] \leftrightarrow \fml[2]$ is \kl{valid}.
Additionally, we say that a \kl(fo){formula} $\fml$ is \intro*\kl(fo){satisfiable} (resp.\ \intro*\kl(fo){finitely satisfiable}) if
$\jump{\fml}^{\model} \neq \emptyset$ holds for some \kl{structure} (resp.\ \kl(structure){finite} \kl{structure}) $\model$.

\begin{remark}\label{remark: first-order logic definition}
  Function and constant symbols can be encoded by predicate symbols with functionality axiom (see, e.g., \cite[Sect.\ 19.4]{boolosComputabilityLogic2007})
  and each predicate symbol (of arbitrary arity) can be encoded by binary predicate symbols (see, e.g., \cite[Lem.\ 21.2 (p.\ 275)]{boolosComputabilityLogic2007}, which translations each atomic formula $a(x_1, \dots, x_k)$ into the formula $\exists z, (\bigwedge_{1 \le j \le k} p_j(z, x_j)) \land a'(z, z)$ where $z$ is a fresh variable,
$p_1, \dots, p_k$ are fresh binary symbols for expressing projections, and
$a'$ is a fresh binary symbol for expressing the relation $a$, respectively.
  For instance, the formula $a(x, y, x) \land a(y, y, x)$ is translated into $(\exists z, p_1(z, x) \land p_2(z, y) \land p_3(z, x) \land a'(z, z)) \land (\exists z, p_1(z, y) \land p_2(z, y) \land p_3(z, x) \land a'(z, z))$).
  Thus, by well-known facts, we can give a linear-size translation from formulas of first-order logic with predicate and function symbols of arbitrary arity into $\mathrm{FO}_{=}$ \kl{formulas} (above) preserving \kl{validity} and \kl{finite validity}.
  Here, the \kl(fo){size} of a $k$-ary atomic formula is defined as $\| a(x_1, \dots, x_k) \| = 1 + k$.
  (Note that the translation is not linear-size when the size is defined as $\|a(x_1, \dots, x_k)\| = 1$ and $k$ is not bounded.
  Hence, this linearity depends on the size definition.)
  Hence, we consider the $\mathrm{FO}_{=}$ above (equality $=$ can also be eliminated, see, e.g., \cite[Sect.\ 19.4]{boolosComputabilityLogic2007}, but we introduce it only for convenience).
\end{remark}
\subsection{The calculus of relations}
Let $\sig$ be a countably infinite set of \intro*\kl(cor){(term) variables}.
The set of \intro*\kl(cor){terms} in \intro*\kl{the calculus of relations} (CoR) is defined by:
\begin{align*}
  \term[1], \term[2], \term[3] & \quad::=\quad a \mid \const{I} \mid \term[2]^{-} \mid \term[2] \cap \term[3] \mid \term[2] \cdot \term[3]\mid \term[2]^{\smile} \tag{$a \in \sig$}
\end{align*}
We write $\bigcap \Gamma$ for the \kl(cor){term} $\term_1 \cap \dots \cap \term_n$
where $\Gamma = \set{\term_1, \dots, \term_n}$ is a finite set (and $\term_1, \dots, \term_n$ are ordered by a total order).
We use parentheses in ambiguous situations and use the following notations:
\begin{align*}
  \term[1] \cup \term[2] & \quad\defeq\quad (\term[1]^{-} \cap \term[2]^{-})^{-} & \term[1] \dagger \term[2] & \quad\defeq\quad (\term[1]^{-} \cdot \term[2]^{-})^{-} \\
  \top                   & \quad\defeq\quad \const{I} \cup \const{I}^{-}         & \bot                      & \quad\defeq\quad \top^{-}
\end{align*}
The \intro*\kl(cor){size} $\|\term\| \in \nat$ of a \kl(cor){term} $\term$ is defined by:
\begin{align*}
  \|a\|                       & \defeq 1                                &
  \|\const{I}\|               & \defeq 1                                &
  \|\term[2]^{-}\|            & \defeq 1 + \|\term[2]\|                   \\
  \|\term[2] \cap \term[3]\|  & \defeq 1 + \|\term[2]\| + \|\term[3]\|  &
  \|\term[2] \cdot \term[3]\| & \defeq  1 + \|\term[2]\| + \|\term[3]\| &
  \|\term[2]^{\smile}\|       & \defeq 1 + \|\term[2]\|
\end{align*}

The \intro*\kl(cor){semantics} $\jump{\term}^{\model} \subseteq \domain{\model}^{2}$ of a \kl(cor){term} $\term$ over a \kl{structure} $\model$ over $\sig$ is defined by:
\begin{align*}
  \jump{a}^{\model}                       & \defeq \set{\tuple{v, v'} \in \domain{\model}^2 \mid \tuple{v, v'} \in a^{\model}}                                                                                                                 \\
  \jump{\const{I}}^{\model}               & \defeq \set{\tuple{v, v'} \in \domain{\model}^2 \mid v = v'}                                                                                                                                       \\
  \jump{\term[2]^{-}}^{\model}            & \defeq \domain{\model}^{2} \setminus \jump{\term[2]}^{\model}                                                                                                                                      \\
  \jump{\term[2] \cap \term[3]}^{\model}  & \defeq \jump{\term[2]}^{\model} \cap \jump{\term[3]}^{\model}                                                                                                                                      \\
  \jump{\term[2] \cdot \term[3]}^{\model} & \defeq \set{\tuple{v, v'} \in \domain{\model}^2 \mid \mbox{ for some $v'' \in \domain{\model}$, $\tuple{v, v''} \in \jump{\term[2]}^{\model}$ and $\tuple{v'', v'} \in \jump{\term[3]}^{\model}$}} \\
  \jump{\term[2]^{\smile}}^{\model}       & \defeq \set{\tuple{v, v'} \in \domain{\model}^2 \mid \tuple{v', v} \in \jump{\term[2]}^{\model} }
\end{align*}
We say that a CoR \kl{term} $\term[1]$ and an $\mathrm{FO}_{=}$ \kl{formula} $\fml$ with two distinct free \kl{variables} $x_1$ and $x_2$ are \intro*\kl{semantically equivalent w.r.t.\ binary relations} if
$\jump{\term}^{\model} = \set{\tuple{f(x_1), f(x_2)} \mid f \in \jump{\fml}^{\model}}$ holds for all \kl{structures} $\model$.
It is well-known that we can translate CoR \kl{terms} into $\mathrm{FO3}_{=}$ \kl{formulas}.
\begin{proposition}[the \intro*\kl{standard translation} theorem \cite{tarskiCalculusRelations1941}]\label{proposition: standard translation}
  Let $x_1$ and $x_2$ be distinct \kl{variables}.
  There is a linear-size translation from CoR \kl(cor){terms} into $\mathrm{FO3}_{=}$ \kl(fo){formulas} with two free \kl{variables} $x_1$ and $x_2$ preserving the \kl{semantic equivalence w.r.t.\ binary relations}.
\end{proposition}
\begin{proofsketch}
  Because we can express each operations in CoR by using $\mathrm{FO3}_{=}$ \kl{formulas} (see also \cite[Fig. 1]{nakamuraExpressivePowerSuccinctness2020}).
\end{proofsketch}

Moreover, the set of \intro*\kl(cor){quantifier-free formulas} in CoR is inductively defined as follows:
\begin{align*}
  \fml[1], \fml[2], \fml[3] & \quad::=\quad \term[1] = \term[2] \mid \lnot \fml[2] \mid \fml[2] \land \fml[3]  \tag{$\term[1], \term[2]$ are \kl{terms} in CoR}
\end{align*}
We say that $\term[1] = \term[2]$ is an \intro*\kl(cor){equation}.
An \intro*\kl(cor){inequation} $\term[1] \le \term[2]$ is an abbreviation of the \kl(cor){equation} $\term[1] \cup \term[2] = \term[2]$.
As with \Cref{section: first-order logic}, we use the following notations:
\begin{align*}
  \fml[1] \lor \fml[2] & \;\defeq\; \lnot ((\lnot \fml[1]) \land (\lnot \fml[2])) & 
  \fml[1] \to \fml[2] & \;\defeq\; (\lnot \fml[1]) \lor \fml[2] &
  \const{t}            & \;\defeq\; \const{I} = \const{I}                            \\
  & &
  \fml[1] \leftrightarrow \fml[2] & \;\defeq\; (\fml[1] \to \fml[2]) \land (\fml[2] \to \fml[1]) &
   \const{f}          & \;\defeq\; \lnot \const{t}
\end{align*}

The \intro*\kl(cor formula){size} $\|\fml\| \in \nat$ of a \kl(cor){quantifier-free formula} $\fml$ is defined by:
\begin{align*}
  \|\term[1] = \term[2]\|   & \defeq 1 + \|\term[1]\| + \|\term[2]\| &
  \|\lnot \fml[2]\|         & \defeq 1 + \|\fml[2]\|                 &
  \|\fml[2] \land \fml[3]\| & \defeq 1 + \|\fml[2]\| + \|\fml[3]\|
\end{align*}

The \intro*\kl(cor formula){semantic relation} $\model \models \fml$, where $\fml$ is a \kl(cor){quantifier-free formula} and $\model$ is a \kl{structure} over $\sig$, is defined by:
\begin{align*}
  \model \models \term[1] = \term[2]   & \quad\defiff\quad                                                                                                                                           \jump{\term[1]}^{\model} = \jump{\term[2]}^{\model} \\
  \model \models \lnot \fml[2]         & \quad\defiff\quad \mbox{not ($\model \models \fml[2]$)}                                                                                                                                                           \\
  \model \models \fml[2] \land \fml[3] & \quad\defiff\quad (\model \models \fml[2]) \mbox{ and } (\model \models \fml[3])
\end{align*}
For a \kl(cor){quantifier-free formula} $\fml$ and a \kl{structure} $\model$,
we say that $\fml$ is \intro*\kl(cor){true} on $\model$ if $\model \models \fml$.
Similarly for $\mathrm{FO}_{=}$ \kl{formulas},
we say that a \kl(cor){quantifier-free formula} $\fml$ is \emph{\kl{valid}} (resp.\ \emph{\kl{finitely valid}}) if $\model \models \fml$ holds for all \kl{structures} (resp.\ all \kl(structure){finite} \kl{structures}) $\model$.
We say that two \kl(cor){quantifier-free formulas} $\fml[1], \fml[2]$ are \emph{\kl{semantically equivalent}} if the \kl(cor){quantifier-free formula} $\fml[1] \leftrightarrow \fml[2]$ is \kl{valid}.

It is also well-known that we can translate CoR \kl{quantifier-free formulas} into CoR \kl{equations}, preserving the \kl{semantic equivalence}.
\begin{proposition}[{\intro*\kl{Schr{\"o}der-Tarski translation} theorem \cite{tarskiCalculusRelations1941}}]\label{proposition: Schroder-Tarski}
  There is a linear-size translation from a given \kl(cor){quantifier-free formula} $\fml$ in CoR into a \kl(cor){term} $\term$ such that $\fml$ and $(\term = \top)$ are \kl{semantically equivalent}.
\end{proposition}
\begin{proof}[The proof is from \cite{tarskiCalculusRelations1941}.]
  First, by using $(\term[2] = \term[3]) \leftrightarrow ((\term[2] \cap \term[3]) \cup (\term[2]^{-} \cap \term[3]^{-}) = \top)$,
  we translate a given \kl(cor){quantifier-free formula} into a \kl(cor){quantifier-free formula} s.t.\ each \kl(cor){equation} is of the form $\term[1] = \top$.
  Second, by using the following two \kl{semantic equivalences}, we eliminate logical connectives:
  \begin{align*}
    \lnot (\term[2] = \top)                   & \leftrightarrow \top \cdot \term[2]^{-} \cdot \top = \top &
    (\term[2] = \top) \land (\term[3] = \top) & \leftrightarrow \term[2] \cap \term[3] = \top
  \end{align*}
  We then have obtained the desired \kl(cor){equation} of the form $\term[1] = \top$.
\end{proof}

\begin{remark}\label{remark: semantics equivalence}
  There is also a translation from $\mathrm{FO3}_{=}$ \kl{formulas} (with two free \kl{variables}) into CoR \kl{terms} preserving the \kl{semantic equivalence w.r.t.\ binary relations} (i.e., the converse direction of Prop.~\ref{proposition: standard translation}) \cite{tarskiFormalizationSetTheory1987,givantCalculusRelations2017,nakamuraExpressivePowerSuccinctness2020,nakamuraExpressivePowerSuccinctness2022}, but the best known translation is an exponential-size translation
  and it is open whether there is a subexponential-size translation \cite{nakamuraExpressivePowerSuccinctness2020,nakamuraExpressivePowerSuccinctness2022}.
  This paper's translation given in Sect.\ \ref{section: reduction} only preserves \kl{validity} and \kl{finite validity} and does not preserve the \kl{semantic equivalence w.r.t.\  binary relations}, but it is a linear-size translation.
\end{remark}

\section{A translation from first-order logic into CoR}\label{section: reduction}
We consider the following \kl{structure} transformation.
Based on this transformation, we will give a translation from $\mathrm{FO}_{=}$ \kl{formulas} into $\mathrm{CoR}$ \kl{equations}.
\begin{definition}[\intro*\kl{$k$-tuple structure}]\label{definition: k-extended structure}
  Let $\model$ be a \kl{structure} over $\sig$.
  For $k \ge 1$, the \kl{$k$-tuple structure} of $\model$, written $\model^{(k)}$, is the \kl{structure} over $\sig^{(k)} \defeq \sig \cup \set{U} \cup \set{\pi_{i}, Q_{i}, E_{[1, i]}, E_{[i, k]} \mid 1 \le i \le k}$ defined as follows:
  \begin{align*}
    \domain{\model^{(k)}}      & = \domain{\model}^{k}                                                                                                                                                              \\
    a^{\model^{(k)}}           & = \set{\tuple{\tuple{v, \dots, v}, \tuple{w, \dots, w}} \mid \tuple{v, w} \in a^{\model}} \mbox{ for $a \in \sig$}                                                                 \\
    U^{\model^{(k)}}           & = \set{\tuple{\tuple{v, \dots, v}, \tuple{v, \dots, v}} \mid v \in \domain{\model}}                                                                                                                             \\
     \pi_{i}^{\model^{(k)}}     & = \set{\tuple{\tuple{v_1, \dots, v_i, \dots, v_k}, \tuple{v_i, \dots, v_i}} \mid v_1, \dots, v_k \in \domain{\model}}                                                                          \\
    Q_{i}^{\model^{(k)}}       & = \set{\tuple{\tuple{v_1, \dots, v_k}, \tuple{v_1', \dots, v_k'}} \in \domain{\model}^{k} \times \domain{\model}^{k} \mid v_j = v_j' \mbox{ for $1 \le j \le k$ s.t.\ $j \neq i$}}\\
    E_{[i, i']}^{\model^{(k)}} & = \set{\tuple{\tuple{v_1, \dots, v_k}, \tuple{v_1', \dots, v_k'}} \in \domain{\model}^{k} \times \domain{\model}^{k} \mid v_j = v_j' \mbox{ for $i \le j \le i'$}}                 
  \end{align*}
\end{definition}
Intuitively, in \kl{$k$-tuple structures} $\model^{(k)}$, we reflect each vertex $v$ on $\model$ to the vertex $\tuple{v, \dots, v}$ on $\model^{(k)}$.
The predicate $U$ denotes the set of such vertices on $\model^{(k)}$ (coded into a binary identity relation).
Each $k$-tuple $\tuple{v_1, \dots, v_i, \dots, v_k}$ denotes the values of $k$ \kl{variables}.
By using the predicates $\pi_i$, we can map the tuple to the tuple $\tuple{v_i, \dots, v_i}$ (so, each $\pi_i$ behaves as a projection), which is the vertex indicated by the $i$-th variable.
The predicate $Q_{i}$ relates two $k$-tuples if their $j$-th elements are equal except when $j = i$; we will use $Q_{i}$ to denote the existential quantifier ``$\exists x_i$'' where $x_i$ denotes the $i$-th variable.
The predicate $E_{[i,i']}$ relates two $k$-tuples if their $j$-th elements are equal for $i \le j \le i'$; we will use $E_{[i,i']}$ for succinctly defining $Q_{i}$.

We write $k\mbox{-TUPLE}$ for the class of all \kl{$k$-tuple structures}.
Fig.\ \ref{figure: example} gives a graphical example of \kl{$k$-tuple structures} when $k = 2$.
\begin{figure}[b]
  \centering

  \begin{tabular}{ccc}
    \begin{tikzpicture}[baseline = 0.5ex]
      \graph[grow right = 2.5cm, branch down = 2.5cm, nodes={}]{
      {0/{$0$}[draw, circle, mynode],/}-!-{/, 1/{$1$}[draw, circle, mynode]}
      };
      \graph[use existing nodes, edges={color=black, pos = .5, earrow}, edge quotes={fill=white, inner sep=1.pt,font= \scriptsize}]{
      0 ->[bend right = 10, "$a$"] 1;
      1 ->[bend right = 10, "$b$"] 0;
      };
    \end{tikzpicture} & \hspace{3em} &
    \begin{tikzpicture}[remember picture, baseline = 0.5ex]
      \graph[grow right = 2.5cm, branch down = 2.5cm, nodes={mynode}]{
      {00/{$00$}[draw, circle], 01/{$01$}[draw, circle]}-!-{10/{$10$}[draw,circle], 11/{$11$}[draw, circle]}
      };
      \graph[use existing nodes, edges={color=black, pos = .5, earrow}, edge quotes={fill=white, inner sep=1.pt,font= \scriptsize}]{
      00 ->[bend right = 10, "$a$"] 11;
      11 ->[bend right = 10, "$b$"] 00;
      };
    \end{tikzpicture}                                  \\[7ex]
                                                                                                                                   &              &
    \begin{tikzpicture}[remember picture, overlay]
      \graph[use existing nodes, edges={color=black, pos = .5, earrow}, edge quotes={fill=white, inner sep=1.pt,font= \scriptsize}]{
      00 ->["$\pi_1$", loop above] 00;
      01 ->["$\pi_1$", pos =.7] 00;
      10 ->["$\pi_1$", pos =.7] 11;
      11 ->["$\pi_1$", loop below] 11;
      00 ->["$\pi_2$", loop left] 00;
      01 ->["$\pi_2$", pos =.7] 11;
      10 ->["$\pi_2$", pos =.7] 00;
      11 ->["$\pi_2$", loop right] 11;
      };
    \end{tikzpicture}                                  \\[2ex]
    $\model$                                                                                                                       &              & $\model^{(2)}$
  \end{tabular}

  \caption{Example of \kl{$k$-tuple structures} (when $k = 2$)}
  \label{figure: example}
\end{figure}
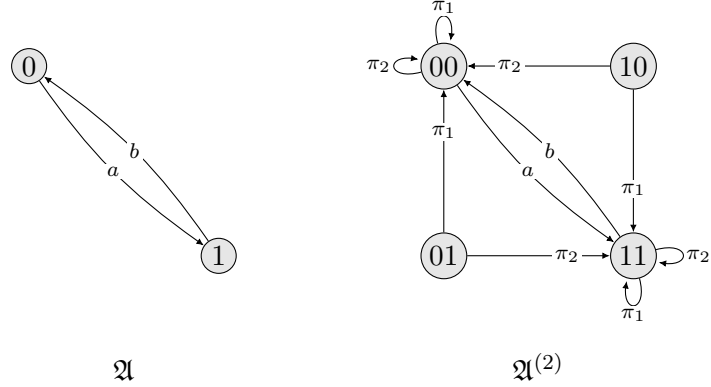
Here, each $2$-tuple $\tuple{v, v'}$ is abbreviated to $v v'$ and the relations of $U$, $Q_i$, $E_{[i, i']}$ on $\model^{(2)}$ are omitted where 
\begin{itemize}
  \item $U^{\model^{(2)}} = \set{\tuple{00, 00}, \tuple{11,11}}$,
  \item $Q_1^{\model^{(2)}} = \set{\tuple{00, 00}, \tuple{00, 10}, \tuple{01, 01}, \tuple{01, 11}, \tuple{10, 00}, \tuple{10, 10}, \tuple{11, 01}, \tuple{11, 11}}$,
  \item $Q_2^{\model^{(2)}} = \set{\tuple{00, 00}, \tuple{00, 01}, \tuple{01, 00}, \tuple{01, 01}, \tuple{10, 10}, \tuple{10, 11}, \tuple{11, 10}, \tuple{11, 11}}$,
  \item $E_{[1,1]}^{\model^{(2)}} = Q_2^{\model^{(2)}}$, $E_{[2,2]}^{\model^{(2)}} = Q_1^{\model^{(2)}}$, and $E_{[1,2]}^{\model^{(2)}} = \jump{\const{I}}^{\model^{(2)}}$.  
\end{itemize}

This construction preserves \kl(structure){finiteness}, so the following holds.
\begin{proposition}\label{proposition: preserving finite}
  For all \kl{structures} $\model$ and $k \ge 1$,
  $\model$ is \kl(structure){finite} if and only if $\model^{(k)}$ is \kl(structure){finite}.
\end{proposition}
Additionally note that in \kl{$k$-tuple structures}, $U, Q_i, E_{[i, i']}$ can be defined by using $\pi_{i}$ as follows:
\begin{align*}
  U           & = \bigcap_{1 \le j \le k} \pi_j                                  &
  E_{[i, i']} & = \bigcap_{i \le j \le i'} \pi_j \cdot \pi_{j}^{\smile}          &
  Q_{i}       & = \bigcap_{1 \le j \le k; j \neq i} \pi_j \cdot \pi_{j}^{\smile}
\end{align*}
Thus, $U, Q_i, E_{[i, i']}$ does not change the expressive power; they are introduced to reduce the output size to be linear.
By using $E_{[1, i - 1]}$ and $E_{[i + 1, k]}$, we can succinctly express $Q_i$ as $Q_i = E_{[1, i - 1]} \cap E_{[i + 1, k]}$.

We show that the class of (the \kl(structure){isomorphism closure} of) \kl{$k$-tuple structures} can be characterized by using \kl{equations} in CoR.
Let $\Gamma^{(k)}$ be the following finite set of \kl{equations} where $i$ ranges over $1 \le i \le k$ and $E_{[1, 0]}$ and $E_{[k+1, k]}$ are the notations for denoting the \kl{term} $\top$:
\begin{align}
   \label{equation: U} & U = \bigcap_{1 \le j \le k} \pi_j  \\
   \label{equation: EL}                                                          & E_{[1, i]} = E_{[1, i - 1]} \cap (\pi_{i} \cdot \pi_{i}^{\smile})                                                 \\
   \label{equation: ER} & E_{[i, k]} = E_{[i + 1, k]}              \cap (\pi_{i} \cdot \pi_{i}^{\smile})                                             \\                                                       \label{equation: Q} & Q_i = E_{[1, i-1]} \cap E_{[i+1, k]}                                                                                                                                                                                                     \\
   \label{equation: U is a subset of I}  & U \le \const{I}                                                                                                                                                                         \\
   \label{equation: pi is functional} & \pi_i^{\smile} \cdot \pi_i \le \const{I}                                                                                                                                                                   \\
   \label{equation: pi is function} & \const{I} \le \pi_i \cdot U \cdot \pi_i^{\smile}                                         \\
   \label{equation: existence} & \top \cdot U \le Q_i \cdot \pi_{i}                                                       \\
   \label{equation: same} & \const{I} = E_{[1, k]}                                                                                                                                      \\
   \label{equation: non-empty} & \top \cdot U \cdot \top = \top                                                                                                            \\
  \label{equation: a U} & a \le U \cdot \top \cdot U
\end{align}
\Cref{equation: U,equation: EL,equation: ER,equation: Q} define $U$, $E_{[i, i']}$, and $Q_i$, respectively.
\Cref{equation: U is a subset of I} expresses that $U$ is a subset of the identity relation.
\Cref{equation: pi is functional,equation: pi is function} express that 
$\pi_i$ is a left-total function into $U$, namely 
$\pi_i$ is a function relation (\Cref{equation: pi is functional} implies that $\pi_i$ is functional and \Cref{equation: pi is function} implies that $\pi_i$ is left-total)
and its range is a subset of $U$.
\Cref{equation: existence} means that the vertex $\tuple{u_1, \dots, u_{i-1}, u_{i}', u_{i+1}, \dots, u_k}$ exists for every vertex $\tuple{u_1, \dots, u_k}$ and every $u_{i}'$ in $U$.
\Cref{equation: same} implies that if each $\pi_j$-images of two vertices are the same, then the vertices themselves are the same.
\Cref{equation: non-empty} expresses that the relation $U$ is not empty.
\Cref{equation: a U} expresses that the domain and the range of the relation $a$ are subsets of $U$.

For a class $\mathcal{C}$ of \kl{structures},
we write $\mathrm{I}(\mathcal{C})$ for the \intro*\kl(structure){isomorphism closure} of $\mathcal{C}$: the minimal class $\mathcal{C}'$ subsuming $\mathcal{C}$ such that, if $\model[1] \in \mathcal{C}$ and $\model[2]$ is \kl(structure){isomorphic} to $\model[1]$, then $\model[2] \in \mathcal{C}'$.
The set $\Gamma^{(k)}$ can characterize the class of the \kl(structure){isomorphism closure} of \kl{$k$-tuple structures} as follows:
\begin{lemma}\label{lemma: characterize}
  Let $\model$ be a \kl{structure} over $\sig^{(k)}$.
  Then we have:
  \[\model \models \bigwedge \Gamma^{(k)} \quad\iff\quad \model \in \mathrm{I}(k\mbox{-TUPLE}).\]
\end{lemma}
\begin{proof}
  ($\Longleftarrow$):
  This direction can be shown by checking that each equation holds on \kl{$k$-tuple structures}.

  ($\Longrightarrow$):
  Let $U_0 = \set{v \mid \tuple{v, v} \in U^{\model}}$.
  Combining $\pi_i^{\smile} \cdot \pi_i \le \const{I}$ (\ref{equation: pi is functional}), $\const{I} \le \pi_i \cdot U \cdot \pi_i^{\smile}$ (\ref{equation: pi is function}), and $U \le \const{I}$ (\ref{equation: U is a subset of I}) yields that
  $\pi_i^{\model}$ is a function from $\domain{\model}$ to $U_0$.
  Let $f \colon \domain{\model} \to U_0^{k}$ be the function defined by $f(v) = \tuple{\pi_1^{\model}(v), \dots, \pi_k^{\model}(v)}$.
  Then $f$ is bijective as follows.
  Let $v_0 \in \domain{\model}$ be an arbitrary vertex.
  Let $w_1, \dots, w_k$ be s.t.\ $f(v_0) = \tuple{w_1, \dots, w_k}$.
  For any $w_1' \in U_0$,
  by $\pi_{1} \cdot \top \cdot U \le Q_1 \cdot \pi_{1}$ (\ref{equation: existence}),
  $\tuple{v_0, w_1} \in \pi_{1}^{\model}$, $\tuple{w_1, w_1'} \in \jump{\top}^{\model}$, and $w_1' \in U_0$,
  there is some $v_1$ such that $\tuple{v_0, v_1} \in Q_1^{\model}$ and $\tuple{v_1, w_1'} \in \pi_{1}^{\model}$.
  Then by $Q_1 = E_{[2, k]}$ (\ref{equation: Q}), we have $f(v_1) = \tuple{w_1', w_2, \dots, w_k}$.
  Similarly, for any $w_2' \in U_0$,
  by $\pi_{2} \cdot \top \cdot U \le Q_2 \cdot \pi_{2}$ (\ref{equation: existence}),
  $\tuple{v_1, w_2} \in \pi_{2}^{\model}$, $\tuple{w_2, w_2'} \in \jump{\top}^{\model}$, and $w_2' \in U_0$,
  there is some $v_2$ such that $\tuple{v_1, v_2} \in Q_2^{\model}$ and $\tuple{v_2, w_2'} \in \pi_{2}^{\model}$.
  Then by $Q_2 = E_{[1, 1]} \cap E_{[3, k]}$ (\ref{equation: Q}), we have $f(v_2) = \tuple{w_1', w_2', w_3, \dots, w_k}$.
  By applying this method iteratively, we have that for any $w_1', \dots, w_k' \in U_0$,
  there is some $v$ such that $f(v) = \tuple{w_1', \dots, w_k'}$.
  Hence, $f$ is surjective.
  Also, if $f(v) = \tuple{w_1, \dots, w_k} = f(v')$, then by $\const{I} = E_{[1, k]}$ (\ref{equation: same}),
  we have $v = v'$.
  Hence, $f$ is injective.
  Therefore, $f$ is bijective.
  
  Note that for each $v \in \domain{\model}$,
  $v \in U_0$ iff
  $\pi_{1}^{\model}(v) = \dots = \pi_{k}^{\model}(v)$ (by $U = \bigcap_{1 \le j \le k} \pi_j$ (\ref{equation: U})) iff
  $f(v) = \tuple{w, \dots, w}$ for some $w$ (by the definition of $f$) iff
  $f(v) = \tuple{v, \dots, v}$ (by $\bigcap_{1 \le j \le k} \pi_j \le \const{I}$ (\ref{equation: U})(\ref{equation: U is a subset of I})).
  Thus, $v \in U_0$ iff $f(v) = \tuple{w, \dots, w}$ for some $w$ iff $f(v) = \tuple{v, \dots, v}$.
  We now define $\model[2]$ as the \kl{structure} over $\sig$, where
  \begin{align*}
    \domain{\model[2]} & = U_0                                  &
    a^{\model[2]}      & = a^{\model} \mbox{ for $a \in \sig$}.
  \end{align*}
  Here, $\model[2]$ is indeed an \kl{structure}, because $U_0$ is not empty by $\top \cdot U \cdot \top = \top$ (\ref{equation: non-empty})
  and $a^{\model[2]} \subseteq U_0^2$ by $a \le U \cdot \top \cdot U$ (\ref{equation: a U}).
  Then the bijection $f$ is an \kl(structure){isomorphism} from $\model[1]$ to $\model[2]^{(k)}$, as follows.
   \begin{itemize}
    \item For $a \in \sig$:
    Let $v, v' \in \domain{\model[1]}$ be arbitrary vertices.
    We distinguish the following cases.
    \begin{itemize}
      \item Case $v \not\in U_0$ or $v' \not\in U_0$:
      By $U \le \const{I}$ (\ref{equation: U is a subset of I}) and $a \le U \cdot \top \cdot U$ (\ref{equation: a U}),
      we have $\tuple{v, v'} \not\in a^{\model[1]}$.
      Also, in this case, $f(v)$ or $f(v')$ is not of the form $\tuple{w, \dots, w}$ for any $w$ (by $v\not\in U_0$ or $v' \not\in U_0$);
      thus by the construction of $\model[2]^{(k)}$, we have $\tuple{f(v), f(v')} \not\in a^{\model[2]^{(k)}}$.
      \item Case $v, v' \in U_0$:
      Then we have $f(v) = \tuple{v, \dots, v}$ and $f(v') = \tuple{v', \dots, v'}$.
      Thus, we have:
      $\tuple{v, v'} \in a^{\model[1]}$ iff $\tuple{v, v'} \in a^{\model[2]}$ iff $\tuple{f(v), f(v')} \in a^{\model[2]^{(k)}}$ (by the construction of $\model[2]^{(k)}$).
    \end{itemize}
    Hence, $\tuple{v, v'} \in a^{\model[1]}$ iff $\tuple{f(v), f(v')} \in a^{\model[2]^{(k)}}$.

    \item For $\pi_i$:
    Let $v, v' \in \domain{\model[1]}$ be arbitrary vertices.
    We distinguish the following cases.
    \begin{itemize}
      \item Case $v' \not\in U_0$:
      By $U \le \const{I}$ (\ref{equation: U is a subset of I}) and $\const{I} \le \pi_i \cdot U \cdot \pi_i^{\smile}$ (\ref{equation: pi is function}),
      we have $\tuple{v, v'} \not\in \pi_i^{\model[1]}$.
      Also, in this case, $f(v')$ is not of the form $\tuple{w, \dots, w}$ for any $w$ (by $v' \not\in U_0$);
      thus by the construction of $\model[2]^{(k)}$, $f(v')$ is not in the range of $\pi^{\model[2]^{(k)}}_{i}$.
      Hence, $\tuple{f(v), f(v')} \not\in \pi_i^{\model[2]^{(k)}}$.
      \item Case $v' \in U_0$:
      Then we have $f(v') = \tuple{v', \dots, v'}$.
      Thus, we have:
      $\tuple{v, v'} \in \pi_i^{\model[1]}$ iff
      the $i$-th element of $f(v)$ is $v'$ (by the definition of $f$) iff
      $\tuple{f(v), \tuple{v', \dots, v'}} \in \pi_i^{\model[2]^{(k)}}$ (by the construction of $\model[2]^{(k)}$) iff
      $\tuple{f(v), f(v')} \in \pi_i^{\model[2]^{(k)}}$ (by $\tuple{v', \dots, v'} = f(v')$).
    \end{itemize}
    Hence, $\tuple{v, v'} \in \pi_i^{\model[1]}$ iff $\tuple{f(v), f(v')} \in \pi_i^{\model[2]^{(k)}}$.

    \item For $U, E_{[i, i']}, Q_i$:
    By \Cref{equation: U,equation: EL,equation: ER,equation: Q} with the fact that $f$ is an \kl(structure){isomorphism} w.r.t.\ $\pi_i$ as above.
  \end{itemize}
  Hence, this completes the proof.
\end{proof}

Using \kl{$k$-tuple structures}, we can give the following translation.
\begin{definition}
  Let $k \ge 1$ and $X = \set{x_1, \dots, x_k}$ where $x_1, \dots, x_k$ are pairwise distinct \kl{variables}.
  For each \kl{formula} $\fml$ of $\mathrm{V}(\fml) \subseteq X$,
  the \kl{term} $\mathrm{T}^{(k)}(\fml)$ is inductively defined as follows:
  \begin{align*}
    \mathrm{T}^{(k)}(a(x_i, x_j))           & \quad\defeq\quad (\pi_{i} \cdot a \cdot \pi_{j}^{\smile}) \cap \const{I}                  \\
    \mathrm{T}^{(k)}(\lnot \fml[2])         & \quad\defeq\quad (\mathrm{T}^{(k)}(\fml[2]))^{-}                                          \\
    \mathrm{T}^{(k)}(\fml[2] \land \fml[3]) & \quad\defeq\quad \mathrm{T}^{(k)}(\fml[2]) \cap \mathrm{T}^{(k)}(\fml[3])                 \\
    \mathrm{T}^{(k)}(\exists x_i, \fml[2])  & \quad\defeq\quad (Q_i \cdot \mathrm{T}^{(k)}(\fml[2]) \cdot Q_i^{\smile}) \cap \const{I}.
  \end{align*}
\end{definition}
\begin{lemma}\label{lemma: equiv}
  Let $k \ge 1$ and $X = \set{x_1, \dots, x_k}$ where $x_1, \dots, x_k$ are pairwise distinct \kl{variables}.
  Let $\model$ be a \kl{structure}.
  For all \kl{formulas} $\fml$ of $\mathrm{V}(\fml) \subseteq X$
  and all $v_1, \dots, v_k \in \domain{\model}$, we have:
  \[\jump{\mathrm{T}^{(k)}(\fml)}^{\model^{(k)}} \quad=\quad \set{\tuple{\tuple{v_1, \dots, v_k}, \tuple{v_1, \dots, v_k}} \mid \set{x_1 \mapsto v_1, \dots, x_k \mapsto v_k} \in \jump{\fml}^{\model} \restriction X}.\]
  Here, $\jump{\fml}^{\model} \restriction X$ denotes the set $\set{f \restriction X \mid f \in \jump{\fml}^{\model}}$ where
  $f \restriction X$ is the restriction of $f$ to $X$.
\end{lemma}
\begin{proof}
  By induction on the structure of $\fml$.

  Case $\fml = a(x_i, x_j)$:
  Since $\mathrm{T}^{(k)}(\fml) = (\pi_{i} \cdot a \cdot \pi_{j}^{\smile}) \cap \const{I}$, we have:
  \begin{align*}
     & \jump{(\pi_{i} \cdot a \cdot \pi_{j}^{\smile}) \cap \const{I}}^{\model^{(k)}}                                                                                                                      \\
     & = \set{\tuple{\tuple{v_1, \dots, v_k}, \tuple{v_1, \dots, v_k}} \mid \tuple{\tuple{v_i, \dots, v_i}, \tuple{v_j, \dots, v_j}} \in a^{\model^{(k)}}}                                                \\
     & = \set{\tuple{\tuple{v_1, \dots, v_k}, \tuple{v_1, \dots, v_k}} \mid \tuple{v_i, v_j} \in a^{\model}}                                                            \tag{Def.\ of $a^{\model^{(k)}}$} \\
     & = \set{\tuple{\tuple{v_1, \dots, v_k}, \tuple{v_1, \dots, v_k}} \mid \set{x_1 \mapsto v_1, \dots, x_k \mapsto v_k} \in \jump{a(x_i, x_j)}^{\model} \restriction X}
  \end{align*}

  Case $\fml = \lnot \fml[2]$:
  Since $\mathrm{T}^{(k)}(\fml) = \mathrm{T}^{(k)}(\fml[2])^{-}$, we have:
  \begin{align*}
     & \jump{\mathrm{T}^{(k)}(\fml[2])^{-}}                                                                                                                                                                        \\
     & = \domain{\model^{(k)}}^{2} \setminus \jump{\mathrm{T}^{(k)}(\fml[2])}                                                                                                                                      \\
     & = \domain{\model^{(k)}}^{2} \setminus \set{\tuple{\tuple{v_1, \dots, v_k}, \tuple{v_1, \dots, v_k}} \mid \set{x_1 \mapsto v_1, \dots, x_k \mapsto v_k} \in \jump{\fml[2]}^{\model} \restriction X} \tag{IH} \\
     & = \set{\tuple{\tuple{v_1, \dots, v_k}, \tuple{v_1, \dots, v_k}} \mid \set{x_1 \mapsto v_1, \dots, x_k \mapsto v_k} \not\in \jump{\fml[2]}^{\model} \restriction X}                                          \\
     & = \set{\tuple{\tuple{v_1, \dots, v_k}, \tuple{v_1, \dots, v_k}} \mid \set{x_1 \mapsto v_1, \dots, x_k \mapsto v_k} \in \jump{\lnot \fml[2]}^{\model} \restriction X}
  \end{align*}

  Case $\fml = \fml[2] \land \fml[3]$:
  Since $\mathrm{T}^{(k)}(\fml) = \mathrm{T}^{(k)}(\fml[2]) \cap \mathrm{T}^{(k)}(\fml[3])$, we have:
  \begin{align*}
     & \jump{\mathrm{T}^{(k)}(\fml[2]) \cap \mathrm{T}^{(k)}(\fml[3])}                                                                                                                                                       \\
     & =  \jump{\mathrm{T}^{(k)}(\fml[2])}^{\model^{(k)}} \cap \jump{\mathrm{T}^{(k)}(\fml[3])}^{\model^{(k)}}                                                                                                               \\
     & = \set{\tuple{\tuple{v_1, \dots, v_k}, \tuple{v_1, \dots, v_k}} \mid \set{x_1 \mapsto v_1, \dots, x_k \mapsto v_k} \in (\jump{\fml[2]}^{\model}\restriction X) \cap (\jump{\fml[2]}^{\model}\restriction X)} \tag{IH} \\
     & = \set{\tuple{\tuple{v_1, \dots, v_k}, \tuple{v_1, \dots, v_k}} \mid \set{x_1 \mapsto v_1, \dots, x_k \mapsto v_k} \in \jump{\fml[2] \land \fml[3]}^{\model}\restriction X}
  \end{align*}

  Case $\fml = \exists x_i, \fml[2]$:
  Since $\mathrm{T}^{(k)}(\fml) = (Q_i \cdot \mathrm{T}^{(k)}(\fml[2]) \cdot Q_i^{\smile}) \cap \const{I}$, we have:
  \begin{align*}
     & \jump{(Q_i \cdot \mathrm{T}^{(k)}(\fml[2]) \cdot Q_i^{\smile}) \cap \const{I}}^{\model^{(k)}}                                                                                                                                                     \\
     & = \set{\tuple{v_1, \dots, v_k}, \tuple{v_1, \dots, v_k} \mid \tuple{\tuple{v_1', \dots, v_k'}, \tuple{v_1', \dots, v_k'}} \in \jump{\mathrm{T}^{(k)}(\fml[2]) }^{\model^{(k)}}                                                                    \\
     & \hspace{7em} \mbox{ for some $v_1', \dots, v_k' \in \domain{\model}$ s.t.\ $v_j' = v_j$ for $1 \le j \le k$ s.t.\ $j \neq i$}} \tag{note that $\jump{\mathrm{T}^{(k)}(\fml[2]) }^{\model^{(k)}} \subseteq \jump{\const{I}}^{\model^{(k)}}$ by IH} \\
     & = \set{\tuple{v_1, \dots, v_k}, \tuple{v_1, \dots, v_k} \mid \set{x_1 \mapsto v_1', \dots, x_k \mapsto v_k'} \in \jump{\fml[2]}^{\model^{(k)}} \restriction X                                                                                     \\
     & \hspace{7em} \mbox{ for some $v_1', \dots, v_k' \in \domain{\model}$ s.t.\ $v_j' = v_j$ for $1 \le j \le k$ s.t.\ $j \neq i$}} \tag{IH}                                                                                                           \\
     & =  \set{\tuple{v_1, \dots, v_k}, \tuple{v_1, \dots, v_k} \mid \set{x_1 \mapsto v_1, \dots, x_k \mapsto v_k} \in \jump{\exists x_i, \fml[2]}^{\model^{(k)}} \restriction X}
  \end{align*}
  Hence, this completes the proof.
\end{proof}
Combining the two above yields the following main lemma.
\begin{lemma}\label{lemma: reduction}
  Let $k \ge 1$ and $X = \set{x_1, \dots, x_k}$ where $x_1, \dots, x_k$ are pairwise distinct \kl{variables}.
  Let $\fml$ be an $\mathrm{FO}_{=}$ \kl{formula} of $\mathrm{V}(\fml) \subseteq X$.
  Then
  \[ \mbox{$(\bigwedge \Gamma^{(k)}) \to \mathrm{T}^{(k)}(\fml) \ge \const{I}$ is [finitely] \kl{valid}} \quad\iff\quad \mbox{$\fml$ is [finitely] \kl{valid}}.\]
\end{lemma}
\begin{shortproof}
  We have:
  \begin{align*}
     & \mbox{$(\bigwedge \Gamma^{(k)}) \to \mathrm{T}^{(k)}(\fml) \ge \const{I}$ is \kl{valid}}                                                                                                                                                                                                                                                                         \\
     & \iff\quad \mbox{$\tuple{v, v} \in \jump{\mathrm{T}^{(k)}(\fml)}^{\model}$ for all $\model$ s.t.\ $\model \models \bigwedge \Gamma^{(k)}$ and all $v \in \domain{\model}$}                                                                                                                                                                                        \\
     & \iff\quad \mbox{$\tuple{\tuple{v_1, \dots, v_k}, \tuple{v_1, \dots, v_k}} \in \jump{\mathrm{T}^{(k)}(\fml)}^{\model[2]^{(k)}}$ for all $\model[2]$ and $v_1, \dots, v_k \in \domain{\model[2]}$}                                                                                                                           \tag{Lem.\ \ref{lemma: characterize}} \\
     & \iff\quad \mbox{$\set{x_1 \mapsto v_1, \dots, x_k \mapsto v_k} \in \jump{\fml}^{\model[2]} \restriction X$ for all $\model[2]$ and $v_1, \dots, v_k \in \domain{\model[2]}$} \tag{Lem.\ \ref{lemma: equiv}}                                                                                                                                                      \\
     & \iff\quad \mbox{$\fml$ is \kl{valid}.}
  \end{align*}
  For \kl{finite validity},
  it is shown in the same way because $\model[2]$ is \kl(structure){finite} iff $\model[2]^{(k)}$ is \kl(structure){finite} (Prop.\ \ref{proposition: preserving finite}).
\end{shortproof}

\begin{theorem}\label{theorem: reduction}
  There is a linear-size translation from $\mathrm{FO}_{=}$ \kl{formulas} into CoR \kl{equations} preserving \kl{validity} and \kl{finite validity}.
\end{theorem}
\begin{shortproof}
  By Lem.\ \ref{lemma: reduction} with the \kl{Schr{\"o}der-Tarski translation} (Prop.\ \ref{proposition: Schroder-Tarski}).
\end{shortproof}

\begin{remark}\label{remark: relation algebra}
  We do not know whether our translation works for the equational theory of (possibly non-representable) relation algebras.
  This is because our construction is not compatible with \emph{quasi-projective relation algebras}---relation algebras having elements $p$ and $q$ s.t.\ $p^{\smile} \cdot p \le \const{I}$, $q^{\smile} \cdot q \le \const{I}$, and $p^{\smile} \cdot q = \top$.
  (As quasi-projective relation algebras are representable \cite{tarskiFormalizationSetTheory1987}, this class is useful to show that a given translation works also for the equational theory of relation algebras, see, e.g., \cite{andrekaReducingFirstorderLogic2013}.)
\end{remark}

\subsection{Reducing to a more restricted syntax of CoR}
We recall that we can eliminate converse $\bl^{\smile}$ and identity $\const{I}$ by using translations given in \cite{nakamuraUndecidabilityFO3Calculus2019}.
\begin{proposition}[{\cite[Lem.\ 7, 9]{nakamuraUndecidabilityFO3Calculus2019}}]\label{proposition: eliminating converse and identity}
  There is a linear-size translation from CoR \kl{equations} into CoR \kl{equations} without $\bl^{\smile}$ nor $\const{I}$ preserving \kl{validity} and \kl{finite validity}.
\end{proposition}
\begin{proposition}[{\cite[Lem.\ 7, 9, 11, 16]{nakamuraUndecidabilityFO3Calculus2019}}]\label{proposition: to one}
  There is a polynomial-size translation from CoR \kl{equations} into CoR \kl{equations} with one variable and without $\bl^{\smile}$ nor $\const{I}$ preserving \kl{validity} and \kl{finite validity}.
\end{proposition}
\begin{remark}
  The translation in \cite[Lem.\ 11]{nakamuraUndecidabilityFO3Calculus2019} (for reducing the number of \kl{variables} to one) is not a linear-size translation, as the output size is not bounded in linear (bounded in quadratic) to the input size.
\end{remark}
By Thm.\ \ref{theorem: reduction} with the two propositions above, we also have the following:
\begin{corollary}\label{corollary: reduction}
  There is a linear-size translation from $\mathrm{FO}_{=}$ \kl{formulas} into CoR \kl{equations} without $\bl^{\smile}$ nor $\const{I}$ preserving \kl{validity} and \kl{finite validity}.
\end{corollary}
\begin{corollary}\label{corollary: reduction to one}
  There is a polynomial-size translation from $\mathrm{FO}_{=}$ \kl{formulas} into CoR \kl{equations} with one variable and without $\bl^{\smile}$ nor $\const{I}$ preserving \kl{validity} and \kl{finite validity}.
\end{corollary}
Additionally, by the \kl{standard translation} (Prop.\ \ref{proposition: standard translation}), we also have obtained the following.
\begin{corollary}\label{corollary: reduction to one FO3}
  There is a polynomial-size translation from $\mathrm{FO}_{=}$ \kl{formulas} into $\mathrm{FO3}$ \kl{formulas} (without equality) with one binary predicate symbol preserving \kl{validity} and \kl{finite validity}.
\end{corollary}

\section{Tseitin translation for CoR}\label{section: Tseitin}
By a similar argument as the \emph{Tseitin translation} \cite{tseitinComplexityDerivationPropositional1983}, which is a translation from propositional formulas into conjunctive normal form preserving \kl(pl){validity} in proposition logic
(see also the Plaisted-Greenbaum translation \cite{plaistedStructurepreservingClauseForm1986} for $\mathrm{FO}_{=}$ and the translation from $\mathrm{FO2}_{=}$ into the Scott class \cite{scottDecisionMethodValidity1962, gradelDecisionProblemTwoVariable1997}),
we can translate into CoR \kl{terms} with bounded alternation of operations.

For each \kl{term} $\term[1]$, we introduce a fresh \kl(cor){variable} $a_{\term[1]}$.
Then for a \kl{term} $\term[1]$, we define the set of \kl{equations} $\Gamma_{\term[1]}$ as follows:
\begin{align*}
  \Gamma_{b}                       & \defeq \set{a_{b} = b}                                                                                              &
  \Gamma_{\term[2]^{-}}            & = \Gamma_{\term[2]} \cup \set{a_{\term[2]^{-}} = a_{\term[2]}^{-}}                                                  &
  \Gamma_{\term[2] \cap \term[3]}  & = \Gamma_{\term[2]} \cup \Gamma_{\term[3]} \cup \set{a_{\term[2] \cap \term[3]} = a_{\term[2]} \cap a_{\term[3]}}     \\
  \Gamma_{\const{I}}               & \defeq \set{a_{\const{I}} = \const{I}}                                                                              &
  \Gamma_{\term[2]^{\smile}}       & = \Gamma_{\term[2]} \cup \set{a_{\term[2]^{\smile}} = a_{\term[2]}^{\smile}}                                        &
  \Gamma_{\term[2] \cdot \term[3]} & = \Gamma_{\term[2]} \cup \Gamma_{\term[3]} \cup \set{a_{\term[2] \cdot \term[3]} = a_{\term[2]} \cdot a_{\term[3]}}
\end{align*}
Then we have the following:
\begin{lemma}\label{lemma: Tseitin transform}
  For all CoR \kl{terms} $\term[1]$, we have:
  \[\mbox{$\term[1] = \top$ is [finitely] \kl{valid}} \quad\iff\quad
    \mbox{$(\bigwedge \Gamma_{\term[1]}) \to a_{\term[1]} = \top$ is [finitely] \kl{valid}.}\]
\end{lemma}
\begin{shortproof}
  For all \kl{structures} $\model$ s.t.\ $\model \models \bigwedge \Gamma_{\term[1]}$,
  we have $\model \models \term[2] = a_{\term[2]}$ for all subterms $\term[2]$ of $\term[1]$, by straightforward induction on $\term[2]$.
  Thus it suffices to prove that $\term[1] = \top$ is [finitely] \kl{valid} $\iff$ $\bigwedge \Gamma_{\term[1]} \to \term[1] = \top$ is [finitely] \kl{valid}.
  Both directions are shown as follows.

  ($\Longrightarrow$):
  Trivial.
  
  ($\Longleftarrow$):
  Since $a_{\term[2]}$ is not occurring in $\term[1]$,
  we can easily transform a \kl{structure} $\model$ s.t.\ $\model \not\models \term[1] = \top$
  into a \kl{structure} $\model'$ s.t.\ $\model' \models \bigwedge \Gamma_{\term[1]}$ and $\model' \not\models \term[1] = \top$ by only modifying $a_{\term[2]}^{\model}$ appropriately.
  Hence this completes the proof.
\end{shortproof}

\begin{example}\label{example: Tseitin Transform}
  The \kl{equation} $((b \cdot c)^{-} \cdot d)^{-} = \top$ is translated into the following \kl{quantifier-free formula} preserving \kl{validity} and \kl{finite validity} (we omit \kl{equations} for each \kl(cor){variables} $b, c, d$, as they are verbose):
  \[\bigwedge \left\{
    \begin{gathered}
      a_{b \cdot c}                      = a_{b} \cdot a_{c},               \quad
      a_{(b \cdot c)^{-}}                = a_{b \cdot c}^{-},                 \\
      a_{(b \cdot c)^{-} \cdot d}        = a_{(b \cdot c)^{-}} \cdot a_{d}, \quad
      a_{((b \cdot c)^{-} \cdot d)^{-}}  = a_{(b \cdot c)^{-} \cdot d}^{-}
    \end{gathered}\right\}
    \to a_{((b \cdot c)^{-} \cdot d)^{-}} = \top\]
  This is \kl{semantically equivalent} to the following \kl{equation}:
  \[(\top \cdot \bigcup \left\{
    \begin{aligned}
      (a_{b \cdot c} \cap (a_{b}^{-} \dagger a_{c}^{-}))                             & ,\  (a_{b \cdot c}^{-} \cap (a_{b} \cdot a_{c})),                                 \\
      (a_{(b \cdot c)^{-}} \cap a_{b \cdot c})                                       & ,\  (a_{(b \cdot c)^{-}}^{-} \cap a_{b \cdot c}^{-}),                              \\
      (a_{(b \cdot c)^{-} \cdot d} \cap (a_{(b \cdot c)^{-}}^{-} \dagger a_{d}^{-})) & ,\  (a_{(b \cdot c)^{-} \cdot d}^{-} \cap (a_{(b \cdot c)^{-}} \cdot a_{d}))   ,  \\
      (a_{((b \cdot c)^{-} \cdot d)^{-}} \cap a_{(b \cdot c)^{-} \cdot d})           & ,\   (a_{((b \cdot c)^{-} \cdot d)^{-}}^{-} \cap a_{(b \cdot c)^{-} \cdot d}^{-})
    \end{aligned}\right\} \cdot \top) \cup a_{((b \cdot c)^{-} \cdot d)^{-}} = \top\]
\end{example}
By using the translation above (and replacing complemented \kl(cor){variables} $b^{-}$ with fresh \kl(cor){variables} $c$ and introducing the axiom $b^{-} = c$), we can translate each CoR \kl{equation} without $\bl^{\smile}$ nor $\const{I}$ into an \kl{equation} of the form $(\top \cdot (\bigcup \Gamma) \cdot \top) \cup a = \top$, where $\Gamma$ is a finite set of \kl{terms} of one of the following forms:
\begin{align*}
   & b \cap c \qquad  b^{-} \cap c^{-} \qquad b \cap (c \dagger d) \qquad b \cap (c \cdot d) \qquad b \cap (c \cap d)
\end{align*}
In this form, the number of alternations of operations, particularly the operations $\cdot$ and $\dagger$ (and similarly, $\cdot$ and $\bl^{-}$), is reduced.
Hence, we have obtained the following:
\begin{theorem}\label{theorem: Tseitin transform}
  There is a linear-size translation from CoR \kl{equations} into \kl{equations} of the form $\term[1] = \top$ preserving \kl{validity} and \kl{finite validity},
  where $\term[1]$ is in the level $\Sigma_{2}^{\mathrm{CoR}}$ of the \intro*\kl{dot-dagger alternation hierarchy} \cite{nakamuraExpressivePowerSuccinctness2022} and $\term$ does not contain $\bl^{\smile}$ nor $\const{I}$.
\end{theorem}
\begin{proof}
  By Prop.\ \ref{proposition: eliminating converse and identity}, there is a linear-size translation from CoR \kl{equations} into CoR \kl{equations} without $\bl^{\smile}$ nor~$\const{I}$.
  Then, by applying the translation of Lem.\ \ref{lemma: Tseitin transform} (with the \kl{Schr{\"o}der-Tarski translation} (Prop.\ \ref{proposition: Schroder-Tarski})) as above, 
  this completes the proof.
\end{proof}
Hence, the equational theory of the form $\term[1] = \top$, where $\term[1]$ is in the level $\Sigma_{2}^{\mathrm{CoR}}$ of the \kl{dot-dagger alternation hierarchy}, is also undecidable, cf.\ \cite[Prop.\ 24]{nakamuraFiniteVariableOccurrenceFragment2023}\cite[Appendix A]{nakamuraFiniteVariableOccurrenceFragment2023arxiv}.

\subsection{Linear-size conservative reduction to G{\"o}del's class $[\forall^3\exists^*, (0, \omega), (0)]$}
Additionally, we note that by using the argument above, we can give a linear-size translation from $\mathrm{FO}_{=}$ \kl{formulas} into $[\forall^3\exists^*, (0, \omega), (0)]$ \kl{sentences} 
(i.e., \kl{sentences} of the form $\forall x, \forall y, \forall z, \exists w_1, \dots, \exists w_n, \fml$ where $n \ge 0$ and $\fml$ is quantifier-free, has only binary predicate symbols and does not have constant symbols, function symbols or non-binary predicate symbols, see e.g., \cite{borgerClassicalDecisionProblem1997} for the notation of the prefix-vocabulary class. $\mathrm{FO}_{=}$ in this paper corresponds to the class $[\mathrm{all}, (0, \omega), (0)]_{=}$) \cite{kurtEntscheidungsproblemLogischenFunktionenkalkiils1933}\cite[p.\ 440]{borgerClassicalDecisionProblem1997} preserving \kl(fo){satisfiability} and \kl(fo){finite satisfiability}.

For example, let us recall the translated \kl{equation} in Example \ref{example: Tseitin Transform}.
By the \kl{standard translation} (Prop.~\ref{proposition: standard translation}), this \kl{equation} is \kl{semantically equivalent w.r.t.\ binary relations} to the following $\mathrm{FO3}$ \kl{sentence} (without equality):
\begin{align*}
   & \exists x, \exists y, \bigvee \left\{ \begin{aligned}
                                              & (a_{b \cdot c}(x, y) \land (\forall w_1, \lnot a_{b}(x, w_1) \lor \lnot a_{c}(w_1, y))),                                                                                     \\
                                              & \quad (\lnot a_{b \cdot c}(x, y) \land (\exists z, a_{b}(x, z) \land a_{c}(z, y))),                                                                                          \\
                                              & a_{(b \cdot c)^{-}}(x, y) \land a_{b \cdot c}(x, y),\  \lnot a_{(b \cdot c)^{-}}(x, y) \land \lnot a_{b \cdot c}(x,y),                                                 \\
                                              & ((a_{(b \cdot c)^{-} \cdot d}(x, y) \land (\forall w_2, \lnot a_{(b \cdot c)^{-}}(x, w_2) \lor \lnot a_{d}(w_2, y))),                                                        \\
                                              & \quad (\lnot a_{(b \cdot c)^{-} \cdot d}(x, y) \land (\exists z, a_{(b \cdot c)^{-}}(x, z) \land a_{d}(z, y)))),                                                             \\
                                              & a_{((b \cdot c)^{-} \cdot d)^{-}}(x, y) \land a_{(b \cdot c)^{-} \cdot d}(x, y),\  \lnot a_{((b \cdot c)^{-} \cdot d)^{-}}(x, y) \land \lnot a_{(b \cdot c)^{-} \cdot d}(x, y)
                                           \end{aligned} \right\} \\
   & \lor (\forall w_3, \forall w_4, a_{((b \cdot c)^{-} \cdot d)^{-}}(w_3, w_4))
\end{align*}

By taking the prenex normal form of the \kl{sentence} above in the ordering of $x, y, z, w_1, w_2, w_3, w_4$, we can obtain an $[\exists^3\forall^{*}, (0, \omega), (0)]$ \kl{sentence} (note that $(\exists z, \fml[2] \lor \fml[3]) \leftrightarrow ((\exists z, \fml[2]) \lor (\exists z, \fml[3]))$).
Thus, as a corollary of Thm.\ \ref{theorem: Tseitin transform}, we can translate CoR \kl{equations} without $\const{I}$ into $[\exists^3\forall^{*}, (0, \omega), (0)]$ \kl{sentences} preserving \kl{validity} and \kl{finite validity}.
Hence, we also have the following:
\begin{corollary}\label{corollary: Tseitin transform}
  There is a linear-size \kl{conservative reduction} from $\mathrm{FO}_{=}$ \kl{formulas} into $[\forall^3\exists^{*}, (0, \omega), (0)]$ \kl{sentences}.
\end{corollary}
\begin{proof}
  By Thm.\ \ref{theorem: Tseitin transform} with the translation above, there is a linear-size translation from $\mathrm{FO}_{=}$ \kl{formulas} into $[\exists^3\forall^{*}, (0, \omega), (0)]$ \kl{sentences} preserving \kl{validity} and \kl{finite validity}.
  Hence, this completes the proof (by considering negated \kl{formulas}).
\end{proof}

\bibliographystyle{fundam}
\bibliography{main-pand}

\appendix
\section{A construction from FO to FO3 not via CoR}\label{section: direct proof}
In the following, we give a direct translation from $\mathrm{FO}_{=}$ \kl{formulas} into $\mathrm{FO3}_{=}$ \kl{formulas} not via CoR (this is almost immediately obtained from Sect.\ \ref{section: reduction} with the \kl{standard translation} from CoR \kl{terms} into $\mathrm{FO3}_{=}$ \kl{formulas} (Prop.\ \ref{proposition: standard translation})).

Let $\Gamma_{\mathrm{FO3}_{=}}^{(k)}$ be the following finite set of $\mathrm{FO3}_{=}$ \kl{formulas} where $i$ ranges over $1 \le i \le k$, $E_{[1, 0]}(x, y)$ and $E_{[k+1, k]}(x, y)$ are the notations for denoting the ``true'' \kl{formula} $\const{t}$ (\Cref{section: first-order logic}), and $x, y, z$ are pairwise distinct \kl{variables}:
\begin{align}
   \label{equation: U'}  & \forall x, \forall y, (U(x, y) \leftrightarrow \bigwedge_{1 \le j \le k} \pi_{j}(x, y))                                                                                                                             \tag{\ref{equation: U}'}                                                              \\
   \label{equation: EL'}  & \forall x, \forall y, (E_{[1, i]}(x,y) \leftrightarrow (E_{[1, i - 1]}(x, y) \land (\exists z, (\pi_{i}(x, z) \land \pi_{i}(y, z)))))                                                                                                                                       \tag{\ref{equation: EL}'}  \\
   \label{equation: ER'}   & \forall x, \forall y, (E_{[i, k]}(x,y) \leftrightarrow (E_{[i + 1, k]}(x, y) \land (\exists z, (\pi_{i}(x, z) \land \pi_{i}(y, z)))))                 \tag{\ref{equation: ER}'}                                                                                                                        \\
   \label{equation: Q'}  & \forall x, \forall y, (Q_i(x, y) \leftrightarrow (E_{[1, i-1]}(x, y) \land E_{[i+1, k]}(x, y)))                                                                                                                                           \tag{\ref{equation: Q}'}                                      \\
   \label{equation: U is a subset of I'}   & \forall x, \forall y, (U(x, y) \to x = y)                                                                        \tag{\ref{equation: U is a subset of I}'}                               \\
   \label{equation: pi is functional'}   & \forall x, \forall y, \forall z, ((\pi_i(x, y) \land \pi_{i}(x, z)) \to y = z) \tag{\ref{equation: pi is functional}'}                                                                                                                                                                    \\
   \label{equation: pi is function'}   & \forall x, \exists y, (\pi_i(x, y) \land U(y, y)) \tag{\ref{equation: pi is function}'}                                                                                                                                        \\
   \label{equation: existence'}   & \forall x, \forall y, (U(y, y) \to (\exists z, Q_i(x, z) \land \pi_i(z, y)))                                                     \tag{\ref{equation: existence}'} \\
   \label{equation: same'}   & \forall x, \forall y, (x = y) \leftrightarrow E_{[1, k]}(x, y) \tag{\ref{equation: same}'}                                                                                                                                      \\
   \label{equation: non-empty'}  & \exists x, U(x, x) \tag{\ref{equation: non-empty}'}                                                                                                                                                                                                                                    \\
   \label{equation: a U'}   & \forall x, \forall y, a(x, y) \rightarrow (U(x, x) \land U(y, y)) \tag{\ref{equation: a U}'}
\end{align}

\begin{lemma}\label{lemma: characterize FO3=}
    Let $\model$ be a \kl{structure} over $\sig^{(k)}$.
    Then
    \[\model \models \bigwedge \Gamma_{\mathrm{FO3}_{=}}^{(k)} \quad\iff\quad \model \in \mathrm{I}(k\mbox{-TUPLE}).\]
\end{lemma}
\begin{proof}
    We can check that $(\bigwedge \Gamma_{\mathrm{FO3}_{=}}^{(k)})$ and $(\bigwedge \Gamma^{(k)})$ are \kl{semantically equivalent}.
    Thus, by Lem.\  \ref{lemma: characterize}, this completes the proof.
\end{proof}

\begin{definition}
    Let $k \ge 3$ and $X = \set{x_1, \dots, x_k}$ where $x_1, \dots, x_k$ are pairwise distinct \kl(fo){variables}.
    For each \kl{formula} $\fml$ of $\mathrm{V}(\fml) \subseteq X$ and $z \in \set{x_1, x_2, x_3}$,
    the FO3 \kl{formula} $\mathrm{T}_{z}^{(k)}(\fml)$ of $\mathrm{V} \subseteq \set{x_1, x_2, x_3}$ is inductively defined as follows, where $z' = \min(\set{x_1, x_2, x_3} \setminus \set{z})$ and $z'' = \min(\set{x_1, x_2, x_3} \setminus \set{z, z'})$ under the ordering $x_1 < x_2 < x_3$:
    \begin{align*}
        \mathrm{T}_{z}^{(k)}(a(x_i, x_j))           & \quad\defeq\quad \exists z', \exists z'', \pi_{i}(z, z') \land a(z', z'') \land \pi_{j}(z, z'') \\
        \mathrm{T}_{z}^{(k)}(\lnot \fml[2])         & \quad\defeq\quad \lnot \mathrm{T}_{z}^{(k)}(\fml[2])                                            \\
        \mathrm{T}_{z}^{(k)}(\fml[2] \land \fml[3]) & \quad\defeq\quad \mathrm{T}_{z}^{(k)}(\fml[2]) \land \mathrm{T}_{z}^{(k)}(\fml[3])              \\
        \mathrm{T}_{z}^{(k)}(\exists x_i, \fml[2])  & \quad\defeq\quad \exists z', Q_i(z, z') \land \mathrm{T}_{z'}^{(k)}(\fml[2]).
    \end{align*}
\end{definition}

\begin{lemma}\label{lemma: equiv FO3=}
    Let $k \ge 3$ and $X = \set{x_1, \dots, x_k}$ where $x_1, \dots, x_k$ are pairwise distinct \kl(fo){variables}.
    Let $\model$ be a \kl{structure}.
    For all \kl{formulas} $\fml$ of $\mathrm{V}(\fml) \subseteq X$,
    all $z \in \set{x_1, x_2, x_3}$,
    and all $u_1, \dots, u_k \in \domain{\model}$, we have:
    \[\set{z \mapsto \tuple{u_1, \dots, u_k}} \in \jump{\mathrm{T}_{z}^{(k)}(\fml)}^{\model^{(k)}} \restriction \set{z} \quad\iff\quad \set{x_1 \mapsto u_1, \dots, x_k \mapsto u_k} \in \jump{\fml}^{\model} \restriction X.\]
\end{lemma}
\begin{shortproof}
    By induction on the structure of $\fml$ (similarly for Lem.\ \ref{lemma: equiv}).
\end{shortproof}
Combining the two above, we have obtained the following main lemma.
\begin{lemma}\label{lemma: reduction FO3=}
    Let $k \ge 3$ and $X = \set{x_1, \dots, x_k}$ where $x_1, \dots, x_k$ are pairwise distinct \kl(fo){variables}.
    Let $\fml$ be a \kl{formula} of $\mathrm{V}(\fml) \subseteq X$ and $z \in \set{x_1, x_2, x_3}$.
    Then
    \[ \mbox{$(\bigwedge \Gamma_{\mathrm{FO3}_{=}}^{(k)}) \to \mathrm{T}_{z}^{(k)}(\fml)$ is [finitely] \kl{valid}} \quad\iff\quad \mbox{$\fml$ is [finitely] \kl{valid}}.\]
\end{lemma}
\begin{shortproof}
    We have:
    \begin{align*}
         & \mbox{$(\bigwedge \Gamma_{\mathrm{FO3}_{=}}^{(k)}) \to \mathrm{T}_{z}^{(k)}(\fml)$ is \kl{valid}}                                                                                                                                                                                                                                                                             \\
         & \iff\quad \mbox{$\set{z \mapsto v} \in \jump{\mathrm{T}_{z}^{(k)}(\fml)}^{\model} \restriction \set{z}$ for all $\model$ s.t.\ $\model \models \bigwedge \Gamma_{\mathrm{FO3}_{=}}^{(k)}$ and all $v \in \domain{\model}$}                                                                                                                                                    \\
         & \iff\quad \mbox{$\set{z \mapsto \tuple{u_1, \dots, u_k}} \in \jump{\mathrm{T}_{z}^{(k)}(\fml)}^{\model[2]^{(k)}} \restriction \set{z}$ for all $\model[2]$ and $u_1, \dots, u_k \in \domain{\model[2]}$}                                                                                                                           \tag{Lem.\ \ref{lemma: characterize FO3=}} \\
         & \iff\quad \mbox{$\set{x_1 \mapsto u_1, \dots, x_k \mapsto u_k} \in \jump{\fml}^{\model[2]} \restriction X$ for all $\model[2]$ and $u_1, \dots, u_k \in \domain{\model[2]}$} \tag{Lem.\ \ref{lemma: equiv FO3=}}                                                                                                                                                              \\
         & \iff\quad \mbox{$\fml$ is \kl{valid}.}
    \end{align*}
    For \kl{finite validity},
    it is shown in the same way, because $\model[2]$ is \kl(structure){finite} iff $\model[2]^{(k)}$ is \kl(structure){finite}.
\end{shortproof}

Additionally, note that equality can be eliminated in $\mathrm{FO3}_{=}$.
\begin{proposition}\label{proposition: equality elimination in FO3}
    There is a linear-size \kl{conservative reduction} from $\mathrm{FO3}_{=}$ \kl{formulas} into $\mathrm{FO3}$ \kl{formulas} without equality.
\end{proposition}
\begin{proofsketch}
    See, e.g., \cite[Prop.\ 19.13]{boolosComputabilityLogic2007} from $\mathrm{FO}_{=}$ to $\mathrm{FO}$.
    This can be proved by replacing each occurrence of equality $=$ with a fresh binary predicate symbol $E$
    and then adding axioms of that $E$ is an equivalence relation
    and of that each binary predicate $a$ satisfies the congruence law w.r.t.\ $E$:
    \[\forall x, \forall x', \forall y, \forall y', (E(x, x') \land E(y, y')) \rightarrow (a(x, y) \leftrightarrow a(x', y')) \]
    While the \kl{formula} above is not in FO3, the construction in \cite[Prop.\ 19.13]{boolosComputabilityLogic2007} still works for FO3 by replacing the \kl{formula} with the conjunction of the following two formulas:
    \begin{align*}
        \forall x, \forall x', \forall y, E(x, x') \rightarrow (a(x, y) \leftrightarrow a(x',y)) \\
        \forall x, \forall y, \forall y', E(y, y') \rightarrow (a(x, y) \leftrightarrow a(x,y'))
    \end{align*}
\end{proofsketch}

\begin{theorem}\label{theorem: reduction FO3=}
    There is a linear-size \kl{conservative reduction} from $\mathrm{FO}_{=}$ \kl{formulas} into $\mathrm{FO3}$ \kl{formulas} (without equality).
\end{theorem}
\begin{shortproof}
    By Lem.\ \ref{lemma: reduction FO3=} with Prop.\ \ref{proposition: equality elimination in FO3}.
\end{shortproof}

\end{document}